\documentclass[journal,11pt,draftcls,onecolumn]{IEEEtran}
\usepackage{amsmath,amssymb,amsbsy,amsthm,geometry}
\usepackage{algorithm} 
\usepackage{algpseudocode} 

\let\epsilon\varepsilon
\let\phi\varphi

\let\epsilon\varepsilon

\newtheorem*{lemma*}{Lemma}
\newtheorem{theorem}{Theorem}

\newtheorem{claim}{Claim}
\newtheorem{definition}{Definition}

\usepackage[usenames]{color}

\begin{document}


\title{Error correction methods based on two-faced processes} 

\author{
 \IEEEauthorblockN{Boris Ryabko\\}
 \IEEEauthorblockA{Federal Research Center for Information and Computational Technologies,
\\Novosibirsk State University, 
and  \\Siberian State University of Telecommunications and Informatics, 
\\
Novosibirsk\\}}
\date{}

\maketitle


\begin{abstract}

A new approach to the problem of error correction in communication channels is proposed, in which the input sequence is transformed in such a way that the interdependence of symbols is significantly increased. Then, after the sequence is transmitted over the channel, this property is used for error correction so that the remaining error rate is significantly reduced. The complexity of encoding and decoding is linear.

\end{abstract}

\section{Introduction} 

Since the publication of Claude Shannon's famous article \cite{sh}, the problem of efficient data transmission over noisy channels has occupied a central place in information theory.  Hundreds of articles are published annually on various error correction methods, among which we have mentioned several books \cite{b1,bb}, where many references can be found. 

Most error correction codes studied in theory and used in practice are based on adding redundant symbols to the transmitted message.
This class of methods includes block codes \cite{b2,b3}, convolutional codes \cite{b4}, hash function-based codes \cite{b5}, and some others.

In this paper, we propose a new approach based on increasing the mutual dependence of message symbols using a special transformation and then using this dependence to correct errors.

\subsection{The main tool of the proposed method: two-faced processes}

In the main part of the article, we will consider an information sources that generate letters from the alphabet $\{0,1\}$ and, in particular, a Bernoulli source $B(p)$ that generates letters with probabilities $P(0) = p, P(1)= q \, (= 1-p), \,\, p \ge 1/2$.

The proposed code will be based on so-called two-faced  processes, which are Markov chains of order (or memory) $s, s \ge 1$. This means that for the generated sequence $... x_{-1} x_0 x_1 ... $, the probability $P(x_{i+1})$ depends on the $s$ preceding symbols, i.e., the probability $P(x_{i+1} = \alpha / ... x_{-1} x_0 x_1 ... x_i)$ is equal to $P(x_{i+1} = \alpha / x_{i-s+1} ... x_i)$, where $\alpha \in \{0,1\}$.
We will describe those chains  by the following transition probability matrices:
 \begin{equation}\label{ch1}
 s=1:    \,\,\,
\begin{array}{c|c|c|}
  probability\, of\, zero \,& p & q \\ 
  \hline 
  after  & 0&1\\
 \end{array} \, .
  \end{equation}
It means, that $Pr\{ x_{i+1} = 0/x_i= 0\} = p $,  $Pr\{ x_{i+1} = 0/x_i= 1\}  $ $= q $ and, of course,   
$Pr\{ x_{i+1} = 1/x_i= 0\} = q $,  $Pr\{ x_{i+1} = 1/x_i= 1\}  $ $= p. $  
It is obvious that $Pr\{ x_{i+1} = 1/...\} = 1 - Pr\{ x_{i+1} =0 /...\} $, and these probabilities are not contained in the matrices.

 \begin{equation}\label{ch2}    For   \,\,
 s=2:   \,\,\,
\begin{array}{c|c|c|c|c|}
  probability\, of\, zero \,& p & q&q&p \\ 
  \hline 
  after  & 00&01&10&11\\
 \end{array} \, .
 \end{equation} 
It means, that $Pr\{ x_{i+1} = 0/x_{i-1}x_i= 00\} = p $,  etc. 

The "typical" sequences for (\ref{ch1}) and  (\ref{ch2}) for large $p$ (say, $0.9$) can be as follows:
 \begin{equation}\label{tt1}
x_1x_2x_3x_4x_5x_6x_7x_8x_9x_{10}  \, = \,
0\,  \, 0\, \, 0\,  \,  0\, 1 \, \, 1 \, 1\, \, 1 \, 1 \,  1 \,\, ,
\end{equation}
 \begin{equation}\label{tt2}  \,\,
x_1x_2x_3x_4x_5x_6x_7x_8x_9x_{10}  \, = \,
0\,  \, 0\, \, 0\,  \,  0 \, 1\, 1 \, 0\,   \, 1\, 1 \, \, 0 \, ,  \, .
\end{equation}(In both cases, a rare event, i.e., with probability $q$, occurred after $x_4$.)

\subsection{ Intuitive description of the ideas underlying the proposed method }

 In this article, we will consider transmission over two binary channels that are very popular in information theory: the erasure channel (EC) and the binary symmetric channel (BSC). For the first channel, any transmitted letter from $\{0,1\}$ is lost, i.e., the receiver receives a special symbol $*$ (instead of the transmitted 0 or 1) with a certain probability $\pi$. For BSC, any letter from the sender can be replaced by another with a certain probability $\pi$, see \cite{co}.

Now let us assume that a binary sequence is transmitted via EC with probability $\pi$, 
that is,  any letter is converted to $*$ during transmission with probability $\pi$.  

 Let us consider three cases: i) the transmitted sequence is generated by a Bernoulli source $B(p), p>1/2$. Obviously, the best strategy is to replace $*$ with $0$, since the probability of 0 is equal to $p>1/2$. Thus, $Pr(error) = q$. In the future, we will use the same strategy for all cases, that is, we will choose the letter with the higher probability and, so,   $Pr(erroe)$ will be equal the smaller probability.

ii) Consider the case where the sequence generated by a two-faced process is transmitted over the   EC.   Suppose that the memory of the process is equal to 1, and the transmitted sequence (\ref{tt1}) is accepted as 
\begin{equation}\label{t1*}   \,
0\,  \, 0\, \, *\,  \,  0\, 1 \, \, 1 \, 1\, \, * \, 1 \,  1 \, . 
\end{equation}
(That is, the third and eighth letters are erased).  Then a simple calculation shows that 
$P(x_3=0/$ $( \ref{t1*} )  ) $ $= P(x_3=0/x_2 =0  \& $( \ref{t1*} ) ) ) $  $ \& $P(x_4=0/x_3 =0 $  $ \& ( \ref{t1*} ) ) ) $
$= p^2/(p^2 + q^2)$. Thus, in this case,  $Pr(error) =$ $P(x_3=1/$ $( \ref{t1*} ) ) $ $= q^2/(p^2 + q^2)$.  The same is true for the eighth letter.

iii) Now suppose that the memory of two-faced process is 2 and sequence (\ref{tt2}) was transmitted via the same EC and the following was received: \begin{equation}\label{t2*}  \,\,
0\,  \, 0\, \, *\,  \,  0 \, 1\, 1 \, 0\,   \, * \, 1 \, \, 0 \, \, .
\end{equation}
The direct calculation of probabilities shows that 
$P(x_3=0/$ $( \ref{t2*} ) = $ 
$ p^3/(p^3 + q^3)$. Thus, in this case,  $Pr(error) $  $ = q^3/(p^3 + q^3)$.  
The same is true for the eighth letter.

So, from these simple examples, we see that the probability of error is $q$ for a Bernoulli process and $q^2/(p^2 + q^2)$, $ q^3/(p^3 + q^3)$ for two-sided processes with memory one and two, respectively. Thus, in this example, we see that increasing the memory of the process reduces the decoding error (in these examples, we do not consider all possible channel error positions, but later we will consider general methods and give general error estimates). Thus, in a sense, increasing the memory  corrects most of the channel errors.

The second ingredient of the proposed error correction method is the availability of a method for transforming the Bernoulli process (and any other random process) into a two-faced process with the necessary memory.

In summary, the proposed method is a two-step procedure: first, the sequence to be transmitted is converted into a two-faced process and transmitted over a noisy channel. Then, the received data undergoes data correction in accordance with the error probability estimates, and the resulting sequence is converted from two-faced back to Bernoulli. 

It turns out that there is a range of parameters where this method significantly reduces the error rate for EC and BSC. It should be noted that the complexity of encoding and decoding this method is linear.

The rest of the article is structured as follows. The next section contains background information on the two-faces processes. The following section describes the proposed method and examines its properties. The third section describes the results of the experiments, and the final section presents a brief conclusion.

\section{Two-faced processes and transformations}

\subsection{two-faced processes}

Let us  define matrices $T_1 =$ $\bigl(\begin{smallmatrix}p&q \\ 0&1\end{smallmatrix} \bigr)$,  $\overline{T}_1 =$ $\bigl(\begin{smallmatrix}q&p \\ 0&1\end{smallmatrix} \bigr)$,  $T_2= T_1\parallel \overline{T}_1= $
$\bigl(\begin{smallmatrix}p&q&q&p \\ 00&01&10&11\end{smallmatrix} \bigr),$
$\overline{T}_2 =$ $\bigl(\begin{smallmatrix}q&p&p&q \\ 00&01&10&11\end{smallmatrix} \bigr),$
$T_3= T_2\parallel \overline{T}_2$,    $...  $  $T_{n+1}= T_n\parallel \overline{T}_n$,  and so on. 
For example, 
$$\overline{T}_3 = \bigl(\begin{smallmatrix}q&p&p&q  & p&q&q&p \\ 000&001&010&011 &100&101&110&111 \end{smallmatrix} \bigr),$$
where, as before,  $Pr\{ x_{i+1} = 0/ x_{i-2}x_{i-1}x_i= 000\} = q $,  etc. 

More precisely, $\overline{T}_k$ is obtained from $T_k$ by changing $p$ and $q$ in 
the top row of $T_k$ and the upper row $T_{k+1}$ is a concatenation of the upper rows $T_k$ and $\overline{T}_k$, while the lower row is a sequence of lexicographically ordered words from $\{0,1\}^{k+1}$.

	These processes were proposed by B. Ryabko \cite{rm} and it was proven that for a process with memory $s$, for any $r \le s$ and words $g$ of length $r$, $Pr\{x_i... x_r = g\} = 2^{-r}$, where the Shannon entropy per letter $h_\infty $ is equal to $p \log_2 p + q \log_2 q $, i.e., $h_\infty $ can be very small (for large $p$). (This is why such processes are called two-faced. If we look at the frequencies of words of length $r, r \le s$, the process appears to be uniformly distributed (i.e., “completely random”), but if we look at the frequencies of long words (those longer than $s$), this distribution is very far from uniform.

\subsection{two-faced transformations}
First, we define two maps  $\phi_{l,w} :$   $\{0,1\}^n$   $\to$  $\{0,1\}^n$  and  $\psi_{w,l} :$    $\{0,1\}^n$   $\to$  $\{0,1\}^n$
where integers $n,l,  (l \le n)$ and a word $w \in \{0,1\}^l$ are parameters. 
 Let there be given  words $w =  w_1 ... , w_l, \in  \{0,1\}^l$, $x_1 ...  x_n \in  \{0,1\}^n$  

The transformation $\phi_{l,w}$: $x \to v$,  is define by equations $v_j= w_{l+j}$ for $j= -l+1, ..., 0$, 
\begin{equation}\label{fi} 
v_{i} =
  \begin{cases} x_{i}  \quad 
   \text{if } \,  \Sigma_{j=i-l}^{i-1}  \,v_j \,\,\text{is\, \, even} 
\\
   \overline{x}_{i}  \quad 
   \text{if } \,  \Sigma_{j=i-l}^{i-1}  \, v_j \,\, \text{is\, odd} \, ,
\end{cases}
\end{equation}
for $i=1, ..., n$ (here and below $\overline{0} = 1,  $ $  \overline{1}=0$).
For some $n$-bit word $v$ the  transformation $\psi_{l,w}$: $v \to u$ is defined by 
\begin{equation}\label{psi} 
u_{i+1} = \begin{cases} v_{i+1}  \quad \text{if } \,  \Sigma_{j=i-l+1}^i  \,v_j \,\, \text{is\, even}    \\   \overline{v}_{i+1}  \quad  \text{if } \,  \Sigma_{j=i-l+1}^i  \, v_j \,\, \text{is\, odd} \, ,\end{cases}\end{equation}for $i=0, ..., n-1.$
For example,  for sequence $x=0000100000$ , $\phi_{1,0}(x) $ equals  $0000111111$, 
$\phi_{2,00}(x) $  equals $0000110110$,  $\psi_{1,0}(0000111111)  = 0000100000$,  
$\psi_{2,00}(0000110110) = 0000100000$.

Note, that $\psi_{l,w}  (\phi_{l,w} (x) ) = x$.

The following definition  plays an important role in the described  codes.  
  \begin{definition}     Let there be  a sequence   $v= v_1 ... v_n$,  some integer $l$ and the word $w \in \{0,1\}^l$. 
 The value $\nu^{i,w}_\tau(v)$ is the number of occurrences of the letter $\tau$ in the word $ \psi_{l,w}(v)$,  where $\tau$ $\in \{0,1\}$.
\end{definition}

The following theorem from \cite{rm}   is true for the transformation  $\phi_{l,w}(x)$:
\begin{theorem} 
Let $l > 0$ be an integer, $w$ be a uniformly distributed binary word of $l$ letters, and let $x=x_1...x_n$ be generated randomly according to a Bernoulli process with parameters $p \in (0,1)$.
Then the sequence $\phi_{l,u}(x)$ obeys distribution 
of   two-faced process  $T_l$ of order $l$ (with matrix $T_l$ ) and probability $p$.\end{theorem}

\section{
The code}
Here we describe the codes with error correction for EC and BSC.
In all cases, we will consider the following scheme. There is an integer $l$ and a word $w \in \{0,1\}^l$, formed according to a uniform distribution, where this word $w$ and   integer $l$ are known to the sender and receiver.
In addition, there is a word $x=x_1 ... x_n$ generated by a Bernoulli source $B(p)$, and this word must be transmitted over a noisy channel. Of course, this word is unknown to the receiver.

Let us denote $\phi_{l,w} (x) = v_1 ... v_n $, 
Define the word $v= w_1 ... w_l  v_1 ... v_n$, that is, $v$ is a word of $(n+l)$ letters, and $ v_{1-l} ... v_0 = w_1 ... w_l$,  
The word $v_1 ... v_n$  is transmitted over a noisy channel, and the receiver receives a distorted word  $v^*_1, ... , v^*_n$.  
Note that in the case of EC,  $v^*_i \in \{0,1,*\}. $      
   It is convenient to define $v^*_{-l+1} ... v^*_0$ as
$ w_1 ... w_l$. (This is possible because the word $w$ is known to the receiver.) Thus, $v^*_{-l+1} ... v^*_0 =$ $v_{-l+1} ... v_0$.
  
Decoding is a two-step procedure: first, the sequence $v^* = v^*_1 v^*_2 ... v^*_n$ is processed into the sequence $v' = v'_1 v'_2 ... v'_n$ in the hope of eliminating some errors (but, at the same time, possibly adding new ones). Then the mapping $\psi_{l,w}(v')$ is applied, and the resulting sequence $x' = x'_1 ... x'_n$ 
is the result of decoding.
Here, the transformation $ v^*_1 v^*_2 ... v^*_n \to $ $v'_1 v'_2 ... v'_n$ plays a key role and will be described separately for EC and BSC.

Another important quantity is the bit error rate (BER), defined as the average frequency of symbols for which $x'_i \ne x_i, i = 1, ... , n$. We denote it as $BER(Code)$, where $Code$ is a specific coding method.

\subsection{Erasure channel}
Suppose that the receiver has received the word $v^* = v^*_1 v^*_2 ... v^*_n$, where all letters belong to the alphabet $\{0,1,*\}$.
The decoder must replace all symbols $*$ with letters from the alphabet $\{0,1\}$ and obtain the word $v' = v'_1 v'_2 ... v'_n$, trying to reduce the number of errors, i.e., minimize the number of letters $v'_1$ for which $v'_i \ne v_i$.  
The method is as follows: for $i=1, ... , n$

i) if the letters $v^*_i \ne *$ , then  
$v'_i = v^*_i$.

ii) $ \,$ if  $v^*_i  =  *$ and  for all letters  $v^*_j \ne *  \,$,    where $ j = i-l +2, ... , i-1, i+1,  ... , i+l+1$,                     
we  define 
$$ \nu' = {{\nu_0^{l, v^*_{i-l} ...  v^*_{i-1} }}} (\psi_{l,   v^*_{i-l} ...  v^*_{i-1} } (1 \, v^*_{i+1}  ... v^*_{i+l} ) )    \, ,
$$
\begin{equation}\label{main} 
v' _{i} = 
  \begin{cases} 1  \quad 
   \text{if } \,  p^{\nu'} (1-p)^{l+1- \nu'}  \ge  (1-p)^{\nu'} p^{l+1- \nu'}  
\\
   0  \quad 
   \text{otherwise} 
\end{cases}
\end{equation}

iii) $ \,$ if  $v^*_i  =  *$ and 
either 
at list for one $j$, $v'_j = *$,  $j \in   \{i-l +2, ... , i-1, i+1,  ... , i+l+1 \} ,$
or $i> n-l$, then
\begin{equation}\label{1/2} 
v' _{i} =
  \begin{cases}   \quad  0 \quad
   \text{with }\, \text{probability} \,  \text{ 1/2}
\\
\quad  1 \quad
   \text{with }\, \text{probability} \,  \text{ 1/2}   
\end{cases} .
\end{equation}
We denote this method as $Code_{EC}$.

Let us consider an example. 
Suppose, 
$ x = 0000100000, \,  Pr(0) = p > 1/2,  \,  l=1,\,  w=0 .
$
Then $$\phi_{1,0}(x) = 0000111111 \,  = v_1 ... v_{10}, \, \, v = v_0v_1 ... v_{10} = 00000111111\,.
$$
Let the word $v_1 ... v_{10} = 0000111111$ be transmitted, and the word 
$v^*_1 ... v^*_{10} = 0*00011*11$ be received.   By definition, $v^* = v^*_0 v^*_1 ... v^*_{10} = 00*00011*11$.
Now we decode $v^*_2=*$ and $v^*_7=*$ according to (\ref{main}),  that is, correct errors. So,  for $v^*_2$ 
we obtain
$$ \nu' = \nu_0^{1,0} (\psi_{1,0}(10)) =  \nu_0^{1,0} (11) = 0\, .
$$
Clearly,   $p^0 (1-p)^{2-0} < (1-p)^0 p^2$ (for $p >1/2$), therefor $v'_2 = 0$ (see (\ref{main})).
Analogically, for $v^*_7=*$, 
$$  \nu' = \nu_0^{1,1} (\psi_{1,1}(11)) =  \nu_0^{1,0} (00) = 2\, 
$$
and $p^2 (1-p)^0 > (1-p)^2 p^0$, hence, $v'_7 =1$. 
So, both errors are corrected and $v'_1 ... v'_{10} = 0000111111 $ $ (= v_1 ... v_{10} )$.  Therefor,  
$\psi_{1,0} (v'_1 ... v'_{10}) = 0000100000,  $ that is, equals an initial sequence.

The second example is for the memory $l=2.$
Let 
$$x = 0000100000,  \, w=00,  \phi_{2,00} = v_1 ... v_{10} = 0000110110, \,  v = v_{-1}v_0 v_1 ... v_{10} = 
00\,  0000110110 \, .$$
Suppose that the word $v_ 1 ... v_{10}$ is transmitted via EC, and the word $000*110*10$ is received.
For $v^*_4$, we obtain from  (\ref{main}) the following.
$$ \nu' = \nu_0^{2,00} (\psi_{2,00}(111) ) = \nu_0^{2,00} (101) = 1\, ; \,\, p^1(1-p)^2 < p^2(1-p)^1 \,
$$ 
So, $v'_4 = 0$. 
Analogically, for $v^*_8=*$.
$$ \nu' = \nu_0^{2,10} (\psi_{2,10}(110) ) = \nu_0^{2,10} (000) = 3\, ; \,\, p^3 >(1-p)^3 \,
$$ 
So, $v'_8 = 1$ and both errors are corrected.

Let us describe properties of  the described code.

\begin{theorem} 
Let there be a Bernoulli source with probability 0 equal to $p, p> 1/2$, generating a sequence of letters $x_1 ... x_n$, which is transmitted via EC with an error probability of $\pi$.  If the $Code_{EC}$ code with parameter $l$ is used to correct errors, then
the bit error rate can be estimated as follows:
\begin{multline}\label{t1} 
BER(Code_{EC}) \le    \pi   \Bigg( \Big(
(1-l/n) (1-\pi)^{2l+1}  \sum_{i=\lfloor (l+1)/2 \rfloor }^l   (^{l+1}_{\,\,i}) p^i (1-p)^{l+1-i}  \Big)  
 \\
+ \Big(  ( (1 - (1- \pi)^{2l+1} ) + l/n) /2   \Big)\Bigg) (l+1).
\end{multline}
\end{theorem}
\begin{proof}
We estimate the probability that a randomly selected value $v'_i \ne v_i$, and then obtain an estimate of $BER$ by multiplying this probability by $l+1$.
Both parts in (\ref{t1}) correspond to options ii) and iii) in the decoding description. More precisely, the first sum corresponds to (\ref{main}), and the next part corresponds to (\ref{1/2}). (Obviously, the first option i) gives $(1-\pi) 0 = 0$.)
\end{proof}
This theorem gives a possibility to find $l$ which minimize (\ref{t1}).

Consider the asymptotic behavior of the estimate (\ref{t1}) for a small channel error probability $\pi$, i.e., $\pi \to 0$. Let $l = (1/\pi)^\gamma,  \gamma \in (0,1),$ and $n >> 1/\pi$. We will use Hoeffding's inequality:
\begin{claim} (Hoeffding's inequality \cite{hoe}). Let the sequence $x_1 ... x_m$ be generated by a Bernoulli source $B(p), p > 1/2$.
Then 
$$\sum_{i=k}^m (^m_i)p^i (1-p)^{m-i} \le
 \exp (- 2m (k/m -p )^2 )      ,  $$
for $k> mp$.
\end{claim}
For $n >1/\pi$,   this claim and known inequality  $ (1- \epsilon)^r \le 1 - r \epsilon $ and $l = (1/\pi)^\gamma$,  
                from (\ref{t1}), we obtain that $BER(Code_{EC}) \le \pi  
(\exp( - const\,\, l) )+ (2l+1) \pi , $ $const = 2 (p-1/2)^2 $.
This inequality shows the effect of the difference $p-1/2$ on the error.  Furthermore, if $l= 1/ \pi^\gamma, \, \gamma \in (0,1)$, then
$= 
 \pi \, \,  o(1) $.

\subsection{BSC channel }
Now we describe a code for BSC channel ($Code_{BSC} )$.
Let, as before, 
$v^* = v^*_1 v^*_2 ... v^*_n$ is obtaind by the receiver  and then processed into the sequence $v' = v'_1 v'_2 ... v'_n$. 
Let us define 
 \begin{equation}\label{del} 
\Delta = \min     \{   s: (1-\pi) p^s (1-p)^{l+1-s} \ge \pi (1-p)^s p^{l+1-s}  \, \,  \& \,\, 1 \le s \le l+1  \}  \, .
\end{equation}
The word $v^* = v^*_1 v^*_2 ... v^*_n$ is processed into the sequence $v' = v'_1 v'_2 ... v'_n$ as follows.
For any $i= 1, ... , n-l$ the decoder calculate
$\nu'_i = {{\nu_{0}^{l, v^*_{i-l} ...  v^*_{i-1} }}} (\psi^{l,   v^*_{i-l} ...  v^*_{i-1} } (v^*_{i}  ... v^*_{i+l} ) )$
and 
\begin{equation}\label{mainsbc} 
v' _{i} = 
  \begin{cases} v^*_i  \quad 
   \text{if } \,  \nu'_i \ge \Delta
\\   \overline{v^*_i }
     \quad 
   \text{otherwise}   \, , 
\end{cases}
\end{equation}
and $v'_i = v^*_i  $ for $i= n-l+1$.

\begin{theorem}  Let there be a Bernoulli source with probability 0 equal to $p, p> 1/2$, generating a sequence of letters $x_1 ... x_n$, which is transmitted via BSC with an error probability of $\pi$.  If the $Code_{BSC}$  with parameter $l$ is used to correct errors, then 
the bit error rate can be estimated as follows:
\begin{equation}\label{t2} 
BER(Code_{BSC}) \le   \Bigg(  \Big(   (1-\pi)  \sum_{i=0}^{\Delta - 1} (^{l+1}_{\,\, i} ) \, p^i (1-p)^{l+1-i}  +
\pi \, \sum_{i=\Delta}^{l} (^{l+1}_{\,\, i} ) \, (1-p)^i  p^{l+1-i}  \,  \Big)\,   \, (1- l/n)  
\end{equation}
$$ + \pi \,  l/n \Bigg) \,\, (l+1)  \, . $$
\end{theorem}
\begin{proof}
Here, the first term corresponds to the case where the transmitted letter was not changed but was incorrectly decoded (i.e., $v_i^* = v_i$, but $v'_i \ne v_i^*$). The second term corresponds to the case where the transmitted letter was changed by the channel, but this error was not corrected.\end{proof}

Applying Hoeffding's inequality (Claim 1), we see that, in general, the asymptotic behavior of the error is close to that of EC.
In particular, the asymptotic behavior of $BER(Code_{BSC})$ at a low channel error probability $\pi$, i.e., $\pi \to 0$, and $l = 1/\pi^\gamma, \gamma \in (0,1)$.
$BER(Code_{BSC}) = \pi o(\pi)$. Thus, the frequency of uncorrected errors tends to 0.

\section{Conclusion}
We can see that encoding and decoding the proposed codes requires a constant number of operations per letter. In a sense, we can say that these methods reduce the number of errors “for free,” since the encoding time of the transmitted sequence and the decoding time of the received sequence are proportional to its length.  

A further development of this research could be the creation of a code that provides the maximum likelihood estimation of the decoded sequence or its simple approximation.

Shannon's famous theorem on information transmission over a noisy channel shows that the codes described can reduce the probability of error to zero only if the entropy of the source is less than the channel capacity (estimates of capacity for both channels can be found in \cite{co}).   If the entropy of the source exceeds the channel capacity, then, generally speaking, the original sequence can be “stretched” by applying known data compression decoders before using the described codes, but this possibility requires further research and, in particular, an assessment of complexity.


\begin{thebibliography}{10}


\bibitem{sh}
Shannon CE. A mathematical theory of communication. The Bell system technical journal. 1948 Jul;27(3):379-423.



\bibitem{b1}
Moon TK. Error correction coding: mathematical methods and algorithms. John Wiley $\& $ Sons; 2020 Dec 15.

\bibitem{bb}
Etzion T. Perfect codes and related structures. 2022 Mar 14.

\bibitem{b2}
Berlekamp ER. The performance of block codes. Notices of the AMS. 2002 Jan;49(1):17-22.



\bibitem{b3}
Moon TK. Error correction coding: mathematical methods and algorithms. John Wiley $\&$ Sons; 2020 Dec 15.


\bibitem{b4}
Johannesson R, Zigangirov KS. Fundamentals of convolutional coding. John Wiley $\&$ Sons; 2015 


\bibitem{b5}
Ryabko B. Linear hash functions and their applications to error detection and correction // Discrete Mathematics, Algorithms and Applications. - 2024. - Vol.16. - Iss. 6. - Art.2350070. -

\bibitem{rm}
Ryabko B. Low-entropy stochastic processes for generating k-distributed and normal sequences, and the relationship of these processes with random number generators. Mathematics. 2019 Sep 10;7(9):838.


\bibitem{co}
T.~M. Cover and J.~A. Thomas, \emph{Elements of information theory}.\hskip 1em
  plus 0.5em minus 0.4em\relax New York, NY, USA: Wiley-Interscience, 2006.




\bibitem{hoe}
Hoeffding, Wassily (1963). "Probability inequalities for sums of bounded random variables" (PDF). Journal of the American Statistical Association. 58 (301): 13–30. doi:10.1080/01621459.1963.10500830









\end{thebibliography}
\end{document}